\documentclass[12pt,a4paper]{article}

\usepackage{amsmath,amsthm,amsfonts,amssymb,colonequals}
\usepackage[latin2]{inputenc}

\providecommand{\npom}{Nešetřil and Ossona de Mendez}
\providecommand{\N}{\ensuremath{\mathbf{N}}}
\providecommand{\R}{\ensuremath{\mathbf{R}} }
\providecommand{\C}{\ensuremath{\mathcal{C}} }
\providecommand{\G}{\ensuremath{\mathcal{G}} }
\providecommand{\ivl}[3]{{#1}_{#2},\ldots,{#1}_{#3}}

\DeclareMathOperator{\ISub}{ISub}
\DeclareMathOperator{\Sub}{Sub}
\DeclareMathOperator{\Hom}{Hom}
\DeclareMathOperator{\AHom}{AHom}

\DeclareMathOperator{\sub}{sub}
\DeclareMathOperator{\isub}{isub}

\DeclareMathOperator{\dom}{dom}

\DeclareMathOperator{\Oh}{\textit{O}}
\newtheorem{thm}{Theorem}

\newtheorem{lemma}[thm]{Lemma}
\newtheorem{obser}[thm]{Observation}
\let\eps\varepsilon
\providecommand{\tnabla}{\buildrel\sim\over\nabla}

\begin{document}

\title{A dynamic data structure for counting subgraphs in sparse graphs\footnote{The work leading to this invention has received funding from the European Research Council under the European Union's Seventh Framework Programme (FP7/2007-2013)/ERC grant agreement no. 259385. The second author received support under project GAUK/592412 of Grant agency of Charles University.}}
\author{Zdeněk Dvořák\thanks{Computer Science Institute, Charles University, Prague, Czech Republic. {\tt rakdver@iuuk.mff.cuni.cz}.}
\and Vojtěch Tůma\thanks{Computer Science Institute, Charles University, Prague, Czech Republic. E-mails: {\tt voyta@iuuk.mff.cuni.cz}.}}

\maketitle

\begin{abstract}
We present a dynamic data structure representing a graph $G$, which allows addition
and removal of edges from $G$ and can determine the number of appearances
of a graph of a bounded size as an induced subgraph of $G$.  The queries are answered in constant time.
When the data structure is used to represent graphs from a class with bounded expansion (which
includes planar graphs and more generally all proper classes closed on topological minors, as well as
many other natural classes of graphs with bounded average degree), the amortized time complexity of updates is
polylogarithmic.
\end{abstract}

\section{Introduction}

The problem of determining whether a graph $H$ is an (induced) subgraph of another graph $G$ is NP-complete in general,
even if $H$ is a clique~\cite{karp}.  Furthermore, it is even $W[1]$-hard when parameterized by $H$~\cite{fellows}, and consequently
it is unlikely to admit an algorithm with time complexity $f(|H|)|V(G)|^{\Oh(1)}$ for any function $f$.  
The best known general algorithms are based on matrix multiplication;
Nešetřil and Poljak~\cite{nepol} gave an $\Oh(|V(G)|^{\omega|V(H)|/3})$-time algorithm, where $\omega$ is the exponent in
the complexity of matrix multiplication.  This was subsequently refined in \cite{klokmil, eigran}.

However, when $G$ is somewhat restricted, the situation changes.  Eppstein~\cite{bib-eppstein99} found an algorithm to decide whether a fixed
graph $H$ is a subgraph of a planar graph $G$ in time $\Oh(|V(G)|)$ (where the multiplicative constant hidden in the $\Oh$-notation depends
on $H$).  The key property of planar graphs used in Eppstein's algorithm is \emph{locally bounded tree-width}---if $v$ is a vertex in a planar
graph $G$ and $G_{v,r}$ is the subgraph of $G$ induced by the vertices at the distance at most $r$ from $v$, then the tree-width of $G_{v,r}$
is bounded by a function of $r$ (Robertson and Seymour~\cite{rs3}).  Frick and Grohe~\cite{fg} proved that this property is sufficient
to decide not only the presence of a subgraph, but also all other properties expressible by a bounded size First Order Logic (FOL) formula, in almost linear time
(i.e., $\Oh(n^{1+\eps})$ for every $\eps>0$).  Let us remark that deciding existence of induced subgraphs is equivalent to deciding existential
FO properties.

Even the assumption of locally bounded tree-width can be relaxed.  \npom{}~\cite{grad1,npom-nd1} introduced \emph{classes of graphs with bounded expansion}
and \emph{nowhere-dense graph classes}; we give the definitions of these notions below.  Here, let us just note that every class with
bounded expansion is nowhere-dense, and that many natural classes of graphs with bounded average degree
(including proper minor-closed classes of graphs, classes of graphs with bounded maximum degree, classes of graphs excluding
a subdivision of a fixed graph, classes of graphs that can be embedded in a fixed surface with bounded number of crossings per each edge and
others, see~\cite{osmenwood}) have bounded expansion.

\npom{}~\cite{npom-old} gave a linear-time algorithm for testing whether a fixed graph is a subgraph of a graph from a class of graphs with bounded expansion.
For nowhere-dense classes, this algorithm runs in almost linear time.  Dvořák, Král' and Thomas~\cite{dkt-fol} extended this result
to all properties expressible in FOL, showing that such properties can be decided in linear time on any class with bounded expansion
(the nowhere-dense case is still open, although the result extends to classes of graphs with \emph{locally bounded expansion}, which generalizes
all previously known results).  Conversely, if a class of graphs $\C$ is closed on subgraphs and it is not nowhere-dense, then the subgraph
problem restricted to $\C$ is $W[1]$-hard (when parameterized by the subgraph).  This shows that the result of \npom{}~\cite{npom-old}
is essentially the best possible.

Dvořák, Král' and Thomas~\cite{dkt-fol} also provided a semidynamic data structure for the problem.  For a fixed first-order formula $\phi$
and a class of graphs $\C$ with bounded expansion, this data structure reprezents a graph $G\in \C$ and can be initialized in time $\Oh(|V(G)|)$.
The data structure enables us to test whether the graph satisfies $\phi$ in constant time.
The graph can be modified by adding and removing edges in constant time, but the edge additions are restricted:
we can only add edges that were removed before.

In this paper, we eliminate the restriction on edge additions; that is, our data structure allows addition of arbitrary edges,
subject to the restriction that the resulting graph still belongs to the considered (bounded expansion or nowhere-dense) class of graphs.
On the other hand, we only handle the case of subgraph testing,
not testing of general FO properties.  We actually deal with the counting version of the problem, i.e., determining how
many times does a fixed graph $H$ appear as an induced subgraph in the represented graph.  Let us now formulate the claim more precisely.

\begin{thm}\label{thm-main}
Let $H$ be a fixed graph and let $\G$ be a class of graphs.
There exists a data structure $\ISub_H(G)$ reprezenting a graph $G\in\G$ which supports the following operations.
\begin{itemize}
\item Determine the number of induced subgraphs of $G$ isomorphic to $H$.
\item Add an edge $e$, i.e., transform $\ISub_H(G)$ to $\ISub_H(G+e)$, under the assumption that $G+e$ is in $\G$. 
\item Delete an edge $e$, i.e., transform $\ISub_H(G)$ to $\ISub_H(G-e)$, under the assumption that $G-e$ is in $\G$.
\end{itemize}
If $\G$ has bounded expansion, then the time complexity of query and edge removal is $\Oh(1)$, while the amortized time complexity of edge addition
is $\Oh(\log^h|V(G)|)$, where $h=\binom{|V(H)|}{2}-1$.
The initialization of the structure can be done in $\Oh(|V(G)|)$ and the space complexity for the structure is $\Oh(|V(G)|)$.
If $\G$ is nowhere-dense, then the time complexity of query is $\Oh(1)$, the amortized time complexity of edge addition or removal is
$\Oh(|V(G)|^\eps)$, the time complexity of the initialization is $\Oh(|V(G)|^{1+\eps})$ and the space complexity is $\Oh(|V(G)|^{1+\eps})$,
for every $\eps>0$.
\end{thm}

A related data structure was previously obtained by Eppstein et al.~\cite{epphin}.
The \emph{h-index} of a graph $G$ is the largest integer $h$ such that $G$ has at least $h$ vertices of degree at least $h$.
Let $\G_h$ denote the class of graphs with h-index at most $h$. 
The data structure of et al.~\cite{epphin} makes it possible to determine the number of all induced subgraphs with at most four vertices in constant time,
with time complexity $\Oh(h^2)$ per modification if it is used to reprezent a graph in $\G_h$.
Note that the class $\G_h$ is closed on topological minors, and thus it has bounded expansion.  Therefore, Theorem~\ref{thm-main} generalizes this result,
but it has somewhat worse time complexity per operation.

The rest of the paper is organized as follows.  First, we give some definitions and auxiliary results needed in the rest of the paper.
Section~\ref{sec-bexp} contains the detailed description of the data structure for induced subgraphs.
Section~\ref{sec:logi} discusses extensions to relational structures and existential FOL properties.

\section{Definitions and auxiliary results}

The graphs considered in this paper are simple, without loops or parallel edges, unless specified otherwise.
For directed graphs, we also do not allow edges joining a single pair of vertices in opposite directions.

The classes of graphs with bounded expansion were introduced by \npom{} in \cite{npom-old}. 
A graph $H$ is said to be a \emph{minor of depth $r$} of a graph $G$, if 
it can be obtained from a subgraph of $G$ by contracting vertex-disjoint subgraphs of radius at most $r$ into single vertices, with arising parallel edges and loops suppressed.
The \emph{Greatest Reduced Average Density} at depth $r$ of graph $G$ then denotes the value
$$\nabla_r(G) = \max\{|E(H)|/|V(H)| : H\mbox{ is a minor of depth $r$ of }G\}.$$
A graph $G$ has \emph{expansion bounded by $f$}, if $f$ is a function from $\N$ to $\R^+$ and $\nabla_r(G) \leq f(r)$ for every $r$.
A class of graphs \G has \emph{bounded expansion}, if there is a function $f$ such that every graph in \G has expansion bounded by $f$.
Let us note that the average degree of a graph $G$ is at most $2\nabla_0(G)$; hence, graphs in any class of graphs with expansion bounded by $f$ have average degree bounded by a constant $2f(0)$.
Similarly, we conclude that every $G\in \G$ has an orientation (even acyclic one) with in-degree at most $D = 2f(0)$.

The nowhere dense classes introduced in~\cite{npom-nd1,npom-nd2} that generalize classes with bounded expansion can be defined in a similar manner---a class \G is \emph{nowhere dense}, if
there is a function $f\colon \N\to \N$ such that the clique number of any minor of depth $r$ of $G$ is at most $f(r)$.
The average degree of graphs with $n$ vertices in a nowhere-dense class is $n^{o(1)}$, i.e., for any nowhere-dense class $\G$ and for every $\eps>0$
there exists a function $g(n)=\Oh(n^\eps)$ such that every graph $G\in \G$ has average degree at most $g(|V(G)|)$. Note that unlike the case
of bounded expansion, the average degree does not have to be bounded by a constant.  More generally, there exists a function $h(n,r)$
such that $h(n,r)=\Oh(n^\eps)$ for every fixed $r$ and the expansion of every graph $G\in\G$ is bounded by the function $f(r)=h(|V(G)|,r)$.

These two concepts of sparsity turned out to be very powerful.
They are very robust, including many widely used classes of sparse graphs, as well as graphs obtained from them by minor perturbations (lexicographical product
with a clique of bounded size, \ldots).
Furthermore, many results regarding colorings, existence of small separators and various kinds of decompositions that hold for specific graph classes (say planar graphs)
generalize to this setting in some form.
Sometimes, using these concepts leads to simpler proofs and algorithms, as they necessarily avoid use of any deep structural theory.
We refer the reader to surveys~\cite{dk-surv,nessurvey} for more information on the subject.

Suppose that $G$ is a directed graph. Vertices $u, v\in V(G)$ form a \emph{fork} if $u$ and $v$ are distinct and non-adjacent and there exists $w\in V(G)$ with $(u,w), (v,w)\in E(G)$.
Let $G'$ be a graph obtained from $G$ by adding the edge $(u,v)$ or $(v,u)$ for every pair of vertices $u$ and $v$ forming a fork.
Then $G'$ is called a \emph{fraternal augmentation} of $G$.  Let us remark that a directed graph can have several different fraternal augmentations,
depending on the choices of directions of newly added edges.
If $G$ has no fork, then $G$ is called \emph{elder graph}.
For an undirected graph $G$, a \emph{$k$-th augmentation of $G$} is a directed graph $G'$ obtained from an orientation of $G$
by iterating fraternal augmentation (for all forks) $k$ times.
Note that $(\binom{|V(G)|}{2} - 2)$-th augmentation of $G$ is an elder graph,
because any graph with at most $1$ edge is already elder and fraternal augmentation of a non-elder graph adds at least one edge.

The following result of \npom~\cite{grad1} shows that fraternal augmentation preserve bounded expansion and nowhere-denseness.
\begin{thm} \label{thm:fraternal_be}
There exist polynomials $f_0$, $f_1$, $f_2$, \ldots with the following property.
Let $G$ be a graph with expansion bounded by a function $g$ and let $G_1$ be an orientation of $G$ with in-degree at most $D$.
If $G'$ is the underlying undirected graph of a fraternal augmentation of $G_1$, then
$G'$ has expansion bounded by the function $g'(r) = f_r(g(2r+1),D)$.
\end{thm}

The fraternal augmentations are a basic tool for deriving properties of graphs with bounded expansion,
e.g., existence of low tree-depth colorings (see \cite{npom-old} for a definition).  Once such a coloring
is found, the subgraph problem can be reduced to graphs with bounded tree-width, where it can be easily
solved in linear time by dynamic programing.  However, we do not know how to maintain a low tree-depth coloring
dynamically (indeed, not even an efficient data structure for maintaining say a proper $1000$-coloring of a planar graph
during edge additions and deletions is known).  The main contribution of this paper is showing that we can count
subgraphs using just the fraternal augmentations, which are much easier to update.

To maintain orientations of a graph, we use the following result by Brodal and Fagerberg \cite{dynds}:
\begin{thm} \label{thm:dynds}
There exists a data structure that, for a graph $G$ with $\nabla_0(G)\le d$, maintains an orientation with maximum in-degree
at most $4d$ within the following bounds:
\begin{itemize}
\item an edge can be added to $G$ (provided that the resulting graph $G'$ still satisfies $\nabla_0(G')\le d$) in
an amortized $\Oh(\log n)$ time, and
\item an edge can be removed in $\Oh(1)$ time, without affecting the orientation of any other edges.
\end{itemize}
The data structure can be initialized in time $O(|V(G)|+|E(G)|)$.
During the updates, the edges whose orientation has changed can be reported in the same time bounds. 
The orientation is maintained explicitly, i.e., each vertex stores a list of in- and out-neighbors.
\end{thm}
Let us remark that the multiplicative constants of the $\Oh$-notation in Theorem~\ref{thm:dynds} do not
depend on $d$, although the implementation of the data structure as described in the paper of Brodal and Fagerberg
requires the knowledge of $d$.

We use this data structure in the following setting.

\begin{thm}\label{thm:aug}
For every $k\ge 0$, there exists an integer $k'$ and a polynomial $g$ with the following property.
Let $\G$ be a class of graphs and $h(n,r)$ a computable function such that the expansion of every graph $G\in \G$
is bounded by $f(r)=h(|V(G)|,r)$.  There exists a data structure representing a $k$-th augmentation $\tilde{G}_k$ of
a graph $G\in G$ with $n$ vertices within the following bounds, where $D=g(h(n,k'))$:
\begin{itemize}
\item the maximum in-degree of $\tilde{G}_k$ is at most $D$,
\item an edge can be added to $G$ (provided that the resulting graph still belongs to $\G$) in
an amortized $\Oh(D\log^{k+1} n)$ time, and
\item an edge can be removed in $\Oh(D)$ time, without affecting the orientation of any other edges.
\end{itemize}
The data structure can be initialized in time $\Oh(Dn+t)$, where $t$ is the time necessary to compute $D$.
The orientation is maintained explicitly, i.e., each vertex stores a list of in- and out-neighbors.
\end{thm}
\begin{proof}
Let $q_0(r)=h(n,r)$.
We use the data structure of Theorem~\ref{thm:dynds} to provide an orientation $G_0$ of $G'_0=G$ with maximum in-degree
at most $4q_0(0)$.  Assume inductively that we have already constructed pairwise edge-disjoint directed graphs
$G_0$, $G_1$, \ldots, $G_i$ with underlying undirected graphs $G'_0$, $G'_1$, \ldots, $G'_i$, such that the expansion of
$\tilde{G}'_i=G'_0\cup \ldots \cup G'_i$ is bounded by a function $q_i(r)$
and $G_i$ has maximum in-degree at most $4q_i(0)$.
We define $G'_{i+1}$ to be the graph with vertex set $V(G)$ and with edges corresponding to the forks in
$\tilde{G}_i=G_0\cup \ldots \cup G_i$.
Note that $\tilde{G}_i$ has maximum in-degree at most $d_i=4(q_0(0)+\ldots+q_i(i))$. By Theorem~\ref{thm:fraternal_be},
$\tilde{G}'_{i+1}=G'_0\cup \ldots \cup G'_{i+1}$ has expansion bounded by $q_{i+1}(r)=f_1(q_i(f_2(r)),d_i)$.
We use the data structure of Theorem~\ref{thm:dynds} to provide an orientation $G_{i+1}$ of $G'_{i+1}$ with
maximum in-degree at most $4q_{i+1}(0)$.

The data structure maintains the orientations $G_0$, $G_1$, \ldots, $G_k$ and their union $\tilde{G}_k$.
Observe that $\tilde{G}_k$ is a $k$-th augmentation of $G$.
Addition of an edge in $G$ may result in change of orientation of $\Oh(\log n)$ edges in $G_0$ (amortized),
which may result in addition or removal of $\Oh(d_0\log n)$ edges in $G'_1$.  Each of them results in
change of orientation of $\Oh(\log n)$ edges in $G_1$ and consequently addition or removal of $\Oh(d_1\log n)$
edges in $G_2$, etc.  Altogether, addition of an edge may result in $\Oh(d_0d_1\ldots d_{k-1}\log^{k+1} n)$ changes,
with the same time complexity.  Similarly, a removal of an edge may result in $\Oh(d_0d_1\ldots d_{k-1})$ changes.

Therefore, Theorem~\ref{thm:aug} holds, since we can choose the integer $k'$ and the polynomial $g$
so that $D\ge \max(d_k,d_0d_1\ldots d_{k-1})$.  Let us remark that we can assume that $D\le n^2$, as otherwise
the claim of the theorem is trivial; hence, the complexity of performing computations
with $D$ (once it was determined during the initialization) does not affect the time complexity of the operations.
\end{proof}

Let $G$ be a directed graph and $S$ a set of its vertices. 
Let $N^+_d(S)$ denote the set of vertices that are reachable from $S$ by a directed path of length at most $d$,
and let $N^+_\infty(S)$ we denote the set of vertices reachable from $S$ by a directed path of any length. 
Similarly, $N^-_d(S)$ and $N^-_\infty(S)$ denote the sets of vertices from that $S$ can be reached by a directed path of length
at most $d$ and by a directed path of any length, respectively.
We also use $N^+_d(v)$, $N^+_\infty(v)$, $N^-_d(v)$, $N^-_\infty(v)$ as shorthands for 
$N^+_d(\{v\})$, $N^+_\infty(\{v\})$, $N^-_d(\{v\})$, $N^-_\infty(\{v\})$, respectively.
We say that a directed graph is \emph{connected} if its underlying undirected graph is connected.
Similarly, \emph{connected components} of a directed graph are its subgraphs induced by vertex sets
of the connected components of its underlying undirected graph.

The key property of elder graphs is that they contain a vertex from that we can reach all other vertices by
directed paths.  Let us prove a stronger claim that we need in the design of our data structure.
A directed tree $T$ with all edges directed away from the root is called an \emph{outbranching}. 
The root of $T$ is denoted by $r(T)$.
Let $H$ be a supergraph of an outbranching $T$ with $V(H)=V(T)$, such that for every edge $(t_1,t_2)\in E(H)$,
there exists a directed path in $T$ either from $t_1$ to $t_2$ or from $t_2$ to $t_1$.
We call such a pair $(H,T)$ a \emph{vineyard}.

\begin{lemma}\label{lem:vine}
If $H$ is a connected elder graph, then there exists an outbranching $T\subseteq H$ such that
$(H,T)$ is a vineyard.
\end{lemma}
\begin{proof}
The claim is obvious if $|V(H)|\le 2$.  Therefore, suppose that $|V(H)|\ge 3$.  By induction,
we can assume that the claim holds for all graphs with less than $|V(H)|$ vertices.
Let $v$ be a vertex of $H$ such that $v$ is not a cutvertex in the underlying undirected
graph of $H$.  Note that every induced subgraph of an elder graph is elder, and thus
by the induction hypothesis, there exists an outbranching $T'$ such that
$(H-v,T')$ is a vineyard.  If $(v,r(T'))\in E(H)$, then we can let $T$ consist of $T'$ and the edge
$(v,r(T'))$.  Therefore, assume that $r(T')\not\in N^+_1(v)$.

Consider a vertex $w\in N^+_1(v)$, and let $r(T')=w_0,w_1,\ldots, w_k=w$ be the directed path in $T$
from $r(T')$ to $w$.  Since both $(w_{k-1},w_k)$ and $(v,w_k)$ are edges of an elder graph $H$,
it follows that either $(w_{k-1},v)$ or $(v,w_{k-1})$ is an edge of $H$.  In latter case,
we can repeat this observation.  Since $r(T')\not\in N^+_1(v)$, we conclude that there exists
$i$ with $0\le i\le k-1$ such that $w_i\in N^-_1(v)$ and $w_j\in N^+_1(v)$ for all $j$ with $i+1\le j\le k$.

In particular, since $G$ is connected, $N^-_1(v)$ is not empty.  
If $u_1$ and $u_2$ are distinct vertices in $N^-_1(v)$, then since $H$ is an elder graph, there
exists an edge joining $u_1$ with $u_2$.  Since $(H,T)$ is a vineyard, there exists a directed
path $Q\subseteq T'$ starting in $r(T')$ such that $N^-_1(v)\subseteq V(Q)$ and the endvertex $z$
of $Q$ belongs to $N^-_1(v)$.

Let $T_1$, \ldots, $T_m$ be all components of $T'-V(Q)$ containing at least one neighbor of $v$.  As we observed
before, we have $(v,r(T_i))\in E(H)$ for $1\le i\le m$.  Let $T$ be the outbranching
obtained from $T'$ by removing the incoming edges of $r(T_1)$, \ldots, $r(T_m)$
and adding the edges $(z,v)$, $(v,r(T_1))$, \ldots, $(v,r(T_m))$.

All neighbors of $v$ belong either to one of the trees $T_1$, \ldots, $T_m$ or to $Q$, and thus they are
joined to $v$ by a directed path in $T$.  Consider an edge $(x,y)\in E(H-v)\setminus E(T)$.  If neither
$x$ nor $y$ belongs to $X=V(T_1)\cup V(T_2)\cup \ldots\cup V(T_m)$, then the path in $T'$ joining $x$ and $y$
also appears in $T$.  If both $x$ and $y$ belong to $X$, then since $(H-v,T')$ is a vineyard, there
exists $i$ (with $1\le i\le m$) such that $x,y\in V(T_i)$, and the path joining $x$ and $y$ in $T'$ also
appears in $T$.  Finally, suppose that say $x$ belongs to $T_1$ and $y$ does not belong to $X$.
Since $(H-v,T')$ is a vineyard, we have $y\in V(Q)$, and $x$ and $y$ are joined in $T$ by the path consisting
of the subpath of $Q$ from $y$ to $z$, the path $zvr(T_1)$ and the path from $r(T_1)$ to $x$ in $T_1$.
Therefore, $(H,T)$ is a vineyard.
\end{proof}

It turns out to be convenient to work with graphs with colored edges.
We do not place any restrictions on the coloring; in particular, edges incident with the same vertex can have the same color.
Suppose that $H$ and $G$ are graphs with colored edges.
A mapping $\phi\colon V(H)\to V(G)$ is a \emph{homomorphism} if for every edge $uv\in E(H)$, we have that $\phi(u)\phi(v)$ is an edge of $G$ of the same
color as $uv$ (and in particular, $\phi(u)\neq \phi(v)$).  A homomorphism is a \emph{subgraph} if it is injective.  It is an \emph{induced subgraph} if it is injective
and $\phi(u)\phi(v)\in E(G)$ implies $uv\in E(H)$, for every $u,v\in V(H)$.
Let $\hom(H,G)$, $\sub(H,G)$ and $\isub(H,G)$ denote the number of homomorphisms, subgraphs and induced subgraphs, respectively, of $H$ in $G$.
Let us note that the definitions of subgraph and induced subgraph distinguish the vertices, i.e., $\sub(H,H)=\isub(H,H)$ is equal to the number of
automorphisms of $H$.

Similarly, if $H$ and $G$ are directed graphs with colored edges, a mapping $\phi\colon V(H)\to V(G)$ is a \emph{homomorphism} if $(u,v)\in E(H)$ implies that $(\phi(u),\phi(v))$
is an edge of $G$ of the same color as $uv$, and $\hom(H,G)$ denotes the number of homomorphisms from $H$ to $G$.

\section{Dynamic data structure for induced subgraphs}\label{sec-bexp}

In this section, we aim to design the data structure $\ISub$ as described in the introduction.  
More precisely, for any positive integer $k$, a fixed graph $H$ with edges colored by colors $\{1,\ldots, k\}$ and a class $\G$ of graphs, we design 
a data structure $\ISub_{H,k}(G)$ representing a graph $G\in \G$ with edges colored by $\{1,\ldots, k\}$, supporting the following operations.

\begin{itemize}
\item Determine $\isub(H,G)$.
\item Change a color of an edge.
\item Add an edge, i.e., transform $\ISub_{H,k}(G)$ to $\ISub_{H,k}(G+e)$, under the assumption that $G+\{e\}$ is in $\G$. 
\item Delete an edge, i.e., transform $\ISub_{H,k}(G)$ to $\ISub_{H,k}(G-e)$. 
\end{itemize}
The complexity of the operations depends on $\G$ and is discussed in more detail in Subsection~\ref{subsec-complex}.
To implement the data structure $\ISub_{H,k}(G)$, we first perform several standard transformations,
reducing the problem to counting homomorphisms.

\subsection{From induced subgraphs to subgraphs}

The data structure $\ISub_{H,k}(G)$ is based on a data structure $\Sub_{H',k}(G)$, which can be used to determine
the number of (not necessarily induced) subgraphs of $H'$ in $G$, i.e., the number $\sub(H',G)$.
The relationship is based on the following claim.

Let $H(+,i,k)$ denote the set of all graphs which can be obtained from $H$ by adding exactly $i$ new edges
and assigning them colors from $\{1,\ldots,k\}$.
\begin{lemma}\label{lemma-ind}
$$\isub(H,G) = \sum_{i=0}^{\binom{|V(H)|}{2} - |E(H)|} (-1)^i \sum_{H'\in H(+,i,k)} \sub(H',G).$$
\end{lemma}
\begin{proof}
Let $\overline{E}$ be the set of all unordered pairs of vertices of $H$ that are not adjacent.
For each pair $uv\in \overline{E}$ such that $uv\not\in E(H)$, let $A_{uv}$ denote
the set of all injective homomorphisms $\phi\colon V(H)\to V(G)$ such that $\phi(u)\phi(v)\in E(G)$.
Observe that $$\isub(H,G)=\sub(H,G)-\left|\bigcup_{uv\in \overline{E}}A_{uv}\right|$$
and that for $1\le i\le |\overline{E}|$,
$$\sum_{H'\in H(+,i,k)} \sub(H',G)=\sum_{X\subseteq \overline{E}, |X|=i} \left|\bigcap_{uv\in X}A_{uv}\right|.$$
The claim of the lemma follows by the principle of inclusion and exclusion.
\end{proof}

The data structure $\ISub_{H,k}(G)$ consists of the collection of the data structures $\Sub_{H',k}(G)$ for all
$H'\in \bigcup_{i=0}^{|\overline{E}|} H(+,i,k)$.  The additions, removals and recolorings of edges of $G$
are performed in all of the data structures, and $\isub(H,G)$ is determined from the queries for $\sub(H',G)$
using the formula from Lemma~\ref{lemma-ind}.  The complexity of each operation with $\ISub_{H,k}(G)$
is thus at most $2^{|V(H)|^2} = \Oh(1)$ times the complexity of the corresponding operation with $\Sub_{H',k}(G)$ for some
graph $H'$ with $|V(H')|=|V(H)|$.

\subsection{From subgraphs to homomorphisms}
Next, we aim to base the data structure $\Sub_{H,k}(G)$ on a data structure $\Hom_{H',k}(G)$,
which counts the number $\hom(H',G)$ of homomorphisms from $H'$ to $G$.  Furthermore, we want
to restrict our attention to the case that $H'$ is connected.

Consider a graph $H$ with colored edges, and let $P$ be a partition of $V(H)$ such that
\begin{itemize}
\item each element of $P$ induces an independent set in $H$, and
\item for every $p_1,p_2\in P$, $u,u'\in p_1$ and $v,v'\in p_2$, if
both $uv$ and $u'v'$ are edges of $H$, then $uv$ and $u'v'$ have the same color.
\end{itemize}
Let $H'$ be the graph obtained from $H$ by identifying the vertices in each part of $P$
and suppressing the parallel edges.  We say that $H'$ is a \emph{projection} of $H$.
Let ${\cal H}^{p}$ denote the set of all projections $H'$ of $H$.

\begin{lemma}\label{lem:hom}
For every graph $H$ with colored edges, there exist integer coefficients $\alpha_{H'}$ such that for every graph $G$ with colored
edges,
$$\sub(H,G) = \sum_{H'\in {\cal H}^{p}} \alpha_{H'} \hom(H',G).$$
\end{lemma}
\begin{proof}
Let $\phi\colon V(H)\to V(G)$ be a homomorphism. Note that $P=\{\phi^{-1}(v):v\in \dom(\phi)\}$ is a partition of $V(H)$
that gives rise to a projection $H'$ of $H$, and $H'$ appears as a subgraph in $G$.
Conversely, if a projection $H'$ of $H$ (given by a partition $P$ of $V(H)$) is a subgraph of $G$, then it corresponds to
a unique homomorphism from $H$ to $G$ that maps all vertices of each element of $P$ to the image of the corresponding vertex of $H'$.

This bijective correspondence shows that
$$\hom(H,G) = \sum_{H'\in {\cal H}^{p}} \sub(H',G).$$
Equivalently,
\begin{equation}\label{eq-homsub}
\sub(H,G) = \hom(H,G) - \sum_{H'\in {\cal H}^{p}\setminus\{H\}} \sub(H',G).
\end{equation}

We prove Lemma~\ref{lem:hom} by induction.  Assume that the claim is true for all graphs with fewer vertices than $H$.
In particular, for every $H'\in {\cal H}^{p}$ other than $H$, there exist coefficients $\alpha^{H'}_{H''}$ such that
$$\sub(H',G)=\sum_{H''\in {\cal H'}^{p}} \alpha^{H'}_{H''}\hom(H'',G)$$ for every graph $G$.
Note that ${\cal H'}^{p}\subseteq {\cal H}^{p}\setminus\{H\}$.
Therefore, Lemma~\ref{lem:hom} follows from (\ref{eq-homsub}) by setting
$\alpha_H=1$ and $$\alpha_{H''}=-\sum_{H'\in {\cal H}^{p}\setminus\{H\}, H''\in {\cal H'}^{p}} \alpha^{H'}_{H''}$$
for every $H''\in {\cal H}^{p}\setminus\{H\}$.
\end{proof}

A similar trick allows us to deal with disconnected graphs.
\begin{obser}\label{lem:con}
Let $H_1$ and $H_2$ be two graphs.
For the disjoint union $H_1\cup H_2$ it holds that
$$\hom(H_1\cup H_2,G) = \hom(H_1,G)\cdot \hom(H_2,G).$$
\end{obser}

In the following subsection, we design a data structure $\Hom_{H,k}(G)$ for a connected graph $H$
with edges colored by $\{1,\ldots,k\}$, which counts the number $\hom(H,G)$ of homomorphisms from $H$ to $G$, and
allows additions, removals and recolorings of edges in $G$.

The data structure $\Sub_{H,k}(G)$ consists of the collection of data structures $\Hom_{H',k}(G)$
for all connected components of projections of $H$.  Edge additions, removals and recolorings in $G$
are performed in all these structures.  The number $\sub(H,G)$ is determined from the queries
to the structures according to the formula following from Lemmas~\ref{lem:hom} and \ref{lem:con}.

The number of projections of $H$ and their components is bounded by a function of $H$,
which we consider to be a constant.  Therefore, the complexity of operations
with $\Sub_{H,k}(G)$ is the same up to a constant multiplicative factor
as the complexity of operations with $\Hom_{H',k}(G)$ with $|V(H')|\le |V(H)|$.

\subsection{Augmented graphs}\label{sec:aug}

In order to implement the data structure $\Hom_{H,k}(G)$, we use fraternal augmentations.
Essentially, we would like to find a bijection between homomorphisms from $H$ to $G$
and between homomorphisms from all possible $h$-th augmentations of $H$ to an $h$-th
augmentation of $G$, where $h=\binom{|V(H)|}{2}-2$.  However, it turns out that we need
to be a bit more careful.

For a graph $F$ with edges colored by colors $\{1,\ldots, k\}$, we define the color of an edge $(u,v)$
of a $t$-th augmentation of $F$ to be the same as the color of $uv$ if $uv\in E(F)$, and to be $0$ otherwise
(i.e., we introduce a new color for the edges added through the fraternal augmentation).
If $F'$ and $F''$ are directed graphs with edges colored by colors $\{0, 1,\ldots, k\}$,
we say that $F''$ is \emph{obtained from $F'$ by recoloring zeros} if $F'$ and $F''$ differ only in
the colors of edges whose color in $F'$ is $0$.

\begin{lemma}\label{lem:hom-ed}
Let $H$ and $G$ be graphs with edges colored by $\{1,\ldots, k\}$ and let $h\ge 0$ be an integer.
Let $\phi\colon V(H)\to V(G)$ be a homomorphism and let $G'$ be an $h$-th augmentation of $G$.
There exists a graph $H'$ obtained from an $h$-th augmentation of $H$ by recoloring zeros,
such that for every edge $(u,v)\in E(H')$,
\begin{itemize}
\item if $\phi(u)\neq \phi(v)$, then $(\phi(u),\phi(v))\in E(G')$, and $(u,v)$ has the same color as $(\phi(u),\phi(v))$; and,
\item if $\phi(u)=\phi(v)$, then the color of $(u,v)$ is $0$.
\end{itemize}
\end{lemma}
\begin{proof}
We prove the claim by the induction on $h$.  If $h=0$,
we let $H'$ be the orientation of $H$ such that each edge $uv\in E(H)$ is oriented towards $v$ if $(\phi(u),\phi(v))\in G'$ and towards $u$
otherwise (i.e., if $(\phi(v),\phi(u))\in G'$), with the colors of the edges of $H'$ matching the colors of the corresponding edges of $H$.

Therefore, suppose that $h>0$.  Let $G_1$ be an $(h-1)$-th augmentation of $G$ such that $G'$ is a fraternal augmentation of $G_1$.
By induction hypothesis, there exists a directed graph $H_1$ obtained from an $(h-1)$-th augmentation $H$ by recoloring zeros,
satisfying the outcome of the lemma.

Let $u$ and $v$ be vertices forming a fork in $H_1$, such that $\phi(u)\neq \phi(v)$.  If $(\phi(u),\phi(v))\in E(G_1)$ or $(\phi(v),\phi(u))\in E(G_1)$,
then we choose the orientation and the color of the edge $uv$ in $H'$ correspondingly.  Otherwise, consider a vertex $w$ such that $(u,w), (v,w)\in E(H_1)$,
and note that since $\phi(u)$ is not adjacent to $\phi(v)$ in $G_1$, the induction hypothesis implies that $\phi(u)\neq\phi(w)\neq \phi(v)$
and that $(\phi(u),\phi(w)), (\phi(v),\phi(w))\in E(G_1)$.  It follows that $\phi(u)$ and $\phi(v)$ form a fork in $G_1$, and thus
$(\phi(u),\phi(v))\in E(G)$ or $(\phi(v),\phi(u))\in E(G)$.  We choose the orientation of the edge $uv$ in $H'$ correspondingly, and color it by $0$.

Finally, for each pair $u,v\in V(H_1)$ forming a fork in $H_1$ such that $\phi(u)=\phi(v)$, we choose an orientation of $uv$ in $H'$ arbitrarily
and assign it color $0$.  Observe that the fraternal augmentation $H'$ of $H_1$ and its coloring satisfy the outcome of Lemma~\ref{lem:hom-ed} as required.
Furthermore, the choices of colors and orientations of edges of $H'$ that are not mapped to a single vertex are uniquely determined by the conditions of the lemma.
\end{proof}

Lemma~\ref{lem:hom-ed} inspires the following definition.
Let $F'$ be a directed graph with edges colored by $\{0,1,\ldots, k\}$.
Let $P$ be a partition of vertices of $F'$ such that
\begin{itemize}
\item for every $p\in P$, the subgraph of $F'$ induced by $p$ is connected and contains only edges colored by $0$; and
\item if $p_1,p_2\in P$ are distinct, $u,u'\in p_1$, $v,v'\in p_2$ and $(u,v)$ is an edge, then $(v',u')$
is not an edge, and if $(u',v')$ is an edge, then it has the same color as $(u,v)$.
\end{itemize}
Let $F''$ be the directed graph with edges colored by $\{0,1,\ldots, k\}$, such that $V(F'')=P$ and $(p_1,p_2)\in E(F'')$
if and only if $(v_1,v_2)\in E(F')$ for some $v_1\in p_1$ and $v_2\in p_2$; and in this case, $(p_1,p_2)$ and $(v_1,v_2)$ have the
same color.  That is, $F''$ is obtained from $F'$ by identifying the vertices in each part of $P$ and suppressing the parallel edges and loops,
and we also remember which vertices of $F'$ correspond to each vertex of $F''$.
We say that $F''$ is a \emph{$0$-contraction} of $F'$.

We aim to find a bijection between the homomorphisms from an undirected graph $H$ to an undirected graph $G$
and the homomorphisms from all possible $0$-contractions of augmentations of $H$ to a fixed augmentation of $G$.
We will need the following uniqueness result.

\begin{lemma}\label{lem:unique}
Let $H$ be a graph with edges colored by $\{1,\ldots, k\}$ and let $h=\binom{|V(H)|}{2}-2$.
Let $G'$ be a directed graph with edges colored by $\{0,1,\ldots, k\}$, such that
each vertex of $G'$ is contained in a loop with color $0$, but $G'$ has no other loops or parallel edges.
Let $H_1$ and $H_2$ be graphs obtained from $h$-th augmentations of $H$ by recoloring zeros, such that
there exists $\phi\colon V(H)\to V(G')$ which is a homomorphism both from $H_1$ and from $H_2$ to $G'$.
Let $P_0$ be the partition of $V(H)$ such that two vertices $u,v\in V(H)$ belong to the same part in $P$
if and only if $\phi(u)=\phi(v)$.
For $i\in \{1,2\}$, let $P_i$ be the partition of $V(H)$ such that each $p\in P_i$ is the vertex set of a connected
component of the subgraph of $H_i$ induced by vertices in some part $p'\in P_0$.  Let $H'_i$ be the $0$-contraction
of $H_i$ corresponding to $P_i$.  Then $H'_1=H'_2$.
\end{lemma}
\begin{proof}
Before proceeding with the proof, let us remark that the assumption that $\phi$ is a homomorphism
from $H_i$ to $G'$ ensures that the conditions on the partition $P_i$ from the definition of a $0$-contraction
are satisfied.  Furthermore, $H'_1=H'_2$ implies $P_1=P_2$.

Suppose that $F$ is a directed graph with vertex set $V(H)$ and with edges colored by $\{0,1,\ldots, k\}$ such that
$\phi$ is a homomorphism from $F$ to $G'$.  Let $P(F)$ be the partition of $V(H)$
such that each $p\in P(F)$ is the vertex set of a connected
component of the subgraph of $F$ induced by vertices in some part $p'\in P_0$.

Note that both $H_1$ and $H_2$ are elder graphs.
Let $H^0_1$ be an orientation of $H$ and let $H^0_1$, $H^1_1$, \ldots, $H^k_1$ be a sequence of directed graphs
with edges colored by $\{0,1,\ldots, k\}$, such that $H_1=H^k_1$ and for $1\le i\le k$, the graph
$H^i_1$ is obtained from $H^{i-1}_1$ by adding an edge joining two vertices forming a fork.

We are going to construct a sequence $H^0_2$, $H^1_2$, \ldots, $H^k_2$, where
\begin{itemize}
\item $H^0_2$ is an orientation of $H$ and $H^i_2\subseteq H_2$ for $0\le i\le k$,
\item $H^i_2$ is obtained from $H^{i-1}_2$ by repeatedly adding edges
joining two vertices forming a fork, for $1\le i\le k$, and
\item $P(H^i_1)=P(H^i_2)$ and the $0$-contractions of $H^i_1$ and $H^i_2$ corresponding to this
partition are identical, for $0\le i\le k$.
\end{itemize}

We set $H^0_2=H^0_1$.  Since $\phi$ is a homomorphism from both $H_1$ and $H_2$ to $G'$ and no edge
of $H$ is colored by $0$, we have $H^0_2\subseteq H_2$, hence $H^0_2$ satisfies the required properties.

Suppose now that $1\le i\le k$ and that we have already constructed $H^{i-1}_2$.
Let $u,v, w\in V(H)$ be the vertices such that $u$ and $v$ are not adjacent in $H^{i-1}_1$,
$(u,w), (v,w)\in E(H^{i-1}_1)$ and $(u,v)\in E(H^i_1)$.  Let $P_{i-1}=P(H^{i-1}_1)=P(H^{i-1}_2)$.
If $u$ and $v$ belong to the same part of $P_{i-1}$, then note that $P(H^i_1)=P_{i-1}$
and the $0$-contractions of $H^{i-1}_1$ and $H^i_1$ corresponding to this partition are identical.
Therefore, we can set $H^i_2=H^{i-1}_2$.

Suppose now that $u$ and $v$ belong to different parts $p_u,p_v\in P_{i-1}$.
Let $p_w\in P_{i-1}$ be the part containing $w$.  Note that $p_w\neq p_u$, as otherwise
$H_1$ would contain both edges $(u,v)$ and $(v,w)$ with $\phi(u)=\phi(w)\neq \phi(v)$, contrary to the
assumption that $\phi$ is a homomorphism from $H_1$ to $G'$.  If $p_w=p_v$, then note that
$P(H^i_1)=P_{i-1}$ and the corresponding $0$-contractions of $H^{i-1}_1$ and $H^i_1$ are identical,
hence we can set $H^i_2=H^{i-1}_2$.

Therefore, we can assume that $p_u\neq p_w\neq p_v$.  In this case we construct
$H^i_2$ by initially setting $H^i_2\colonequals H^{i-1}_2$ and then adding edges as described in the
rest of this paragraph.
Since there exist edges between $w$ and $u$ and $v$,
we also have $\phi(u)\neq \phi(w)\neq \phi(v)$.  Since the $0$-contractions of $H^{i-1}_1$ and $H^{i-1}_2$
corresponding to $P_{i-1}$ are identical, there exist vertices $u'\in p_u$, $v'\in p_v$ and $w_1,w_2\in p_w$
with $(u',w_1),(v',w_2)\in E(H^{i-1}_2)$.  Let $w_1=x_1$, $x_2$, \ldots, $x_t=w_2$ be an induced
path between $w_1$ and $w_2$ in the underlying undirected graph of the subgraph of $H^{i-1}_2$ induced
by $p_w$.  Let us also set $x_0=u'$ and $x_{t+1}=v'$.  Since $(x_0,x_1),(x_{t+1},x_t)\in E(H^{i-1}_2)$,
observe that there exists $j$ (with $1\le j\le t$)
such that $(x_{j-1},x_j),(x_{j+1},x_j)\in E(H^{i-1}_2)$.  Since $H_2$ is an elder graph, we have
that either $(x_{j-1},x_{j+1})$ or $(x_{j+1},x_{j-1})$ is an edge of $H_2$.  We add this edge to $H^i_2$
and consider the path $x_0$, $x_1$, \ldots, $x_{j-1}$, $x_{j+1}$, \ldots, $x_{t+1}$.  Let us note
that if $j=1$ and $t>1$, then we added the edge $(x_0,x_2)$, as we have $x_0\in p_u$, $x_2\in p_w$,
there already exists an edge from $p_u$ to $p_w$, and $\phi$ is a homomorphism from $H_2$ to $G$
that maps all vertices of $p_u$ to $\phi(u)$ and all vertices of $p_w$ to $\phi(w)$.  A symmetric argument holds in the case that
$j=t>1$. Therefore, we can repeat this procedure until an edge between $u'$ and $v'$ is added.

If $\phi(u)\neq\phi(v)$, then the last added edge is $(u',v')$, since $\phi$ is a homomorphism
from both $H_1$ and $H_2$ to $G$.  Furthermore, all other added edges were inside $p_w$, or between $p_u$ and $p_w$,
or between $p_v$ and $p_w$, hence $P(H^i_2)=P_{i-1}=P(H^i_1)$.  We conclude that the corresponding $0$-contractions
are identical as required.
If $\phi(u)=\phi(v)$, then we similarly conclude that both $P(H^i_1)$ and $P(H^i_2)$ are obtained from $P_{i-1}$
by merging $p_u$ and $p_v$, and that the corresponding $0$-contractions of $H^i_1$ and $H^i_2$ are identical.

Therefore, there exists the sequence $H^0_2$, \ldots, $H^k_2$ with the required properties.
Since $H^k_2\subseteq H_2$, we have that $P(H_1)=P(H^k_2)$ is a refinement of $P(H_2)$.
By switching the role of $H_1$ and $H_2$ in the argument, we conclude that $P(H_2)$ is a refinement of $P(H_1)$.
Therefore, $P(H_1)=P(H_2)$.  Since $H_1=H^k_1$ and $H^k_2$ have the same $0$-contraction corresponding to this
partition, it follows that $H'_1\subseteq H'_2$.  By symmetry, we have $H'_2\subseteq H'_1$, and thus $H'_1=H'_2$.
\end{proof}

Lemmas~\ref{lem:hom-ed} and \ref{lem:unique} enable us to express the number of homomorphisms for undirected
graphs in the terms of the homomorphisms of their augmentations.
For a graph $H$ with edges colored by $\{1,\ldots, k\}$, let ${\cal H}^e$ denote 
the set of all $0$-contractions of graphs obtained by recoloring zeros from $h$-th augmentations of $H$,
where $h=\binom{|V(H)|}{2}-2$.

\begin{lemma}\label{lem:hom-corr}
Let $H$ and $G$ be graphs with edges colored by $\{1,\ldots, k\}$ and let $h=\binom{|V(H)|}{2}-2$.
If $G'$ is an $h$-th augmentation of $G$, then
$$\hom(H,G)=\sum_{H'\in {\cal H}^e} \hom(H',G').$$
\end{lemma}
\begin{proof}
Consider a homomorphism $\phi$ from $H$ to $G$.  Let $H_0$ be a graph obtained from an $h$-th augmentation of $H$ by recoloring zeros
such that $H_0$ satisfies the outcome of Lemma~\ref{lem:hom-ed}.  Let $P_0$ be the partition of $V(H_0)$ such that
two vertices $u,v\in V(H_0)$ are in the same part if and only if $\phi(u)=\phi(v)$.  Note that every $p\in P_0$ induces
a subgraph of $H_0$ whose edges have color $0$.  Let $P$ be the refinement of $P_0$ such that every $p'\in P$ is the vertex
set of a connected component of the underlying undirected graph of the subgraph induced in $H_0$ by some $p\in P_0$.
Let $H'$ be the $0$-contraction corresponding to the partition $P$ (which satisfies the assumptions from the definition of a $0$-contraction
since $\phi$ is a homomorphism).  Let $\phi'\colon V(H')\to V(G)$ be the mapping such that $\phi'(p)=\phi(v)$ for every $p\in V(H')$ and $v\in p$.
Observe that $\phi'$ is a homomorphism from $H'$ to $G'$.  This defines a mapping $\Phi(\phi)=(H',\phi')$ which 
assigns a graph $H'\in {\cal H}^e$ and a homomorphism $\phi'\colon V(H')\to V(G')$ to each homomorphism $\phi\colon V(H)\to V(G)$.

We need to prove that $\Phi$ is a bijection.  Note that if $\Phi(\phi)=(H',\phi')$, then for each $v\in V(H)$, we have $\phi(v)=\phi'(p)$,
where $p$ is the vertex of $H'$ such that $v\in p$.  Therefore, $\Phi$ is an injection, and it suffices to argue that $\Phi$ is surjective.

Consider arbitrary $H'\in{\cal H}^e$ and a homomorphism $\phi'$ from $H'$ to $G'$.  Since $H'\in{\cal H}^e$, there exists a graph $H'_0$ obtained
from an $h$-th augmentation of $H$ by recoloring zeros such that $H'$ is a $0$-contraction of $H'_0$.
Let $\phi\colon V(H)\to V(G)$ be the mapping defined by $\phi(v)=\phi'(p)$, where $p$ is the vertex of $H'$ such that $v\in p$.  Note that if $uv\in E(H)$,
then $(u,v)\in E(H'_0)$ or $(v,u)\in E(H'_0)$ and this edge has nonzero color, and thus there exist distinct vertices $p_u,p_v\in V(H')$ with
$u\in p_u$ and $v\in p_v$ such that an orientation of $p_up_v$ is an edge of $H'$ of the same color.  Since $\phi'$ is a homomorphism,
we conclude that $\phi'(p_u)\phi'(p_v)=\phi(u)\phi(v)$ is an edge of $G$ of the same color.  It follows that $\phi$ is a homomorphism
from $H$ to $G$.

We need to prove that $\Phi(\phi)=(H',\phi')$.  Suppose that $\Phi(\phi)=(H'',\phi'')$.
Let $H_0$ be the graph from the definition of $\Phi(\phi)$.  Since $\phi$ is a homomorphism from both $H_0$ and $H'_0$ to
the graph obtained from $G'$ by adding loops of color $0$ to each vertex,
Lemma~\ref{lem:unique} implies that $H''=H'$.  Since the homomorphism from $H'$ to $G'$ is uniquely determined by $\phi$,
it also follows that $\phi''=\phi'$, as required.
Therefore, $\Phi$ is indeed a bijection, and the equality of the lemma follows.
\end{proof}

Furthermore, taking $0$-contractions preserves elderness.

\begin{lemma}
If $H$ is an elder graph with edges colored by $\{0, 1,\ldots, k\}$
and $H'$ is a $0$-contraction of $H$, then $H'$ is an elder graph.
\end{lemma}
\begin{proof}
Let $P$ be a partition of $V(H)$ that gives rise to $H'$.  Suppose that vertices $u',v'\in V(H')$ form a fork, i.e., they are non-adjacent
and there exists a vertex $w'\in V(H')$ with $(u',w'),(v',w')\in E(H')$.

By the definition of a $0$-contraction,
there exist vertices $u\in u'$, $v\in v'$ and $w_1,w_2\in w'$ such that
$(u,w_1),(v,w_2)\in E(H)$.  Furthermore, the underlying undirected graph of the subgraph of $H$ induced
by $w'$ is connected, hence it contains a path $Q=x_1x_2\ldots x_t$ with $x_1=w_1$ and $x_t=w_2$.
Let us choose the vertices $w_1$, $w_2$ and the path $Q$ so that the length of $Q$ is minimal.

Since $u'$ and $v'$ are not adjacent, it follows that $u$ and $v$ are not adjacent, and since $H$ is an
elder graph, we conclude that $w_1\neq w_2$.  Note that $Q$ is an induced path.  Since $H$ is an elder graph,
it follows that for every $i$ with $2\le i\le t-1$, either $(x_{i-1},x_i)\not\in E(H)$ or $(x_{i+1},x_i)\not\in E(H)$.
This implies that if $(x_1,x_2)\in E(H)$, then $(x_i,x_{i+1})\in E(H)$ for $1\le i\le t-1$.  Therefore,
either $(x_2,w_1)\in E(H)$ or $(x_{t-1},w_2)\in E(H)$.  By symmetry, we assume the former.
Since $H$ is an elder graph, this implies that either $(u,x_2)\in E(H)$ or $(x_2, u)\in E(H)$.
Since $H'$ is a $0$-contraction of $H$ arising from the partition $P$, $(u,w_1)\in E(H)$ and
$w_1,x_2\in P_w$, it follows that the edge between $u$ and $x_2$ cannot be oriented towards $u$,
and thus $(u,x_2)\in E(H)$.  However, the path between $x_2$ and $w_2$ is shorter than $Q$.
Since $Q$ was chosen so that its length is minimal, this is a contradiction.
\end{proof}

In the following subsection, we design a data structure $\AHom_{(H',T'),k,D}(G')$ for an elder vineyard $(H',T')$
and a directed graph $G'$ of maximum in-degree at most $D$, where both $H'$ and $G'$ have edges colored by $\{0,1,\ldots, k\}$.
The data structure $\AHom$
counts the number $\hom(H',G')$ of homomorphisms from $H'$ to $G'$ and
allows additions, removals, reorientations and recolorings of edges in $G'$.

The data structure $\Hom_{H,k}(G)$ consists of
\begin{itemize}
\item an $h$-th augmentation $G'$ of $G$ (where $h=\binom{|V(H)|}{2}-2$) maintained as described in Theorem~\ref{thm:aug}, and
\item the collection of data structures $\AHom_{(H',T'),k,D}(G')$ for each $H'\in{\cal H}^e$,
where $D$ is the bound from Theorem~\ref{thm:aug} for the class of graphs containing $G$
and $T'$ is an outbranching in $H'$ such that $(H',T')$ is an elder vineyard (which exists by Lemma~\ref{lem:vine}).
\end{itemize}
Note that $|{\cal H}^e|$ is bounded by a function of $H$ and $k$ only, and thus its size is constant.

Edge additions and removals are first performed in the data structure representing the $h$-th augmentation $G'$ of $G$.
Each addition results in $\Oh(D\log^{h+1} |V(G)|)$ changes in $G'$, each removal results in $\Oh(D)$ such changes (amortized).
A recoloring in $G$ only affects one edge of $G'$.  In all the cases, the changes of $G'$ are performed in
all the $\AHom$ substructures.  The number $\hom(H,G)$ is determined by summing the results
of the queries to these substructures, as follows from Lemma~\ref{lem:hom-corr}.

\subsection{Homomorphisms of elder graphs}
In this subsection, we describe the data structure $\AHom$, thus finishing the design of the data
structure for subgraphs.

Let $(H,T)$ be an elder vineyard.
A \emph{clan} is a subset $C$ of vertices of $H$ such that $N^+_1(C)=C$
and the subgraph $T'$ of $T$ induced by $C$ is an outbranching. Let $r(C)$
denote the root of this outbranching $T'$.
The \emph{ghosts} of a clan $C$ are the vertices $N^-_1(C) \setminus C$.
\begin{lemma}\label{lem:clans}
Let $(H,T)$ be an elder vineyard.
\begin{enumerate}
\item For every $v\in V(H)$, the set $N^+_{\infty}(v)$ is a clan.
\item The ghosts of a clan $C$ are exactly the vertices in $N^-_1(r(C)) \setminus C$.
\item All ghosts of a clan $C$ are on the path from $r(H)$ to $r(C)$ in $T$.
\end{enumerate}
\end{lemma}
\begin{proof}
Let us prove the claims separately:
\begin{enumerate}
\item Let $C=N^+_{\infty}(v)$. Clearly, $N^+_1(C)=C$.  Note that the subgraph $H[C]$ of $H$ induced by $C$ is connected.
If the subgraph of $T$ induced by $C$ is not an outbranching, then it contains two components
$T_1$ and $T_2$ joined by an edge of $H[C]$.  Observe that no directed path in $T$ contains
a vertex both in $T_1$ and $T_2$.  This contradicts the assumption that $(H,T)$ is a
vineyard.
\item Suppose that $v$ is a ghost of $C$, i.e., there exists an edge $(v,w)\in E(H)$
for some $w\in C$.  Let $w$ be such a vertex whose distance from $r(C)$ in $T$ is minimal.
If $w\neq r(C)$, then consider the in-neighbor $z$ of $w$ in $T$.
Since $H$ is an elder graph, $v$ and $z$ are adjacent in $H$.  Since $v$ does not belong
to $C$, we have $(z,v)\not\in E(H)$, and thus $(v,z)\in E(H)$.  However, the distance
from $r(C)$ to $z$ in $T$ is smaller than the distance to $w$, which is a contradiction.
Therefore, we have $w=r(C)$ as required.
\item This follows from the definition of vineyard.
\end{enumerate}
\end{proof}
The \emph{extended clan} $C^\ast$ for a clan $C$ is obtained from the subgraph of $H$
induced by $C$ and its ghosts by removing the edges joining pairs of ghosts.

Let $G$ be a directed graph and let $(H,T)$ be an elder vineyard,
where the edges of $G$ and $H$ are colored by colors $\{0,1,\ldots, k\}$.
let $C$ be a clan
with ghosts $\ivl{g}{1}{m}$ listed in the increasing order by their distance from $r(C)$ in $T$
and let $v$ and $\ivl{w}{1}{m}$ be (not necessarily distinct) vertices of $G$.
Note that $g_1$ is the in-neighbor of $r(C)$ in $T$.
Let $\hom_{(H,T)}(C,v,\ivl{w}{1}{m},G)$ denote the number of homomorphisms from
$C^\ast$ to $G$ such that $r(C)$ maps to $v$ and $\ivl{g}{1}{m}$ map to $\ivl{w}{1}{m}$ in order.
Let $\hom((H,T),G,v)$ denote the number of homomorphisms from $H$ to $G$ such that
$r(T)$ maps to $v$.

\begin{thm} \label{thm:main}
Let $(H,T)$ be an elder vineyard with edges colored by $\{0,1,\ldots, k\}$
and let $D$ be an integer.
There exists a data structure $\AHom_{(H,T),k,D}(G)$ representing a directed graph $G$
with edges colored by $\{0,1,\ldots, k\}$ and maximum in-degree at most $D$
supporting the following operations in $\Oh(D^{|V(H)|^2})$ time.
\begin{enumerate}
\item Addition of an edge $e$ to $G$ such that the maximum indegree of $G+e$ is at most $D$.
\item Reorientation of an edge in $G$ such that the maximum indegree of the resulting graph is at most $D$.
\item Removal or recoloring of an edge.
\end{enumerate}
The data structure can be used to determine $\hom((H,T),G,v)$ for a vertex $v\in V(G)$, as well as
$\hom(H,G)$, in $\Oh(1)$.
The data structure can be built in time $\Oh(D^{|V(H)|^2+1}|V(G)|)$ and has space complexity $\Oh(D^{|V(H)|}|V(G)|)$.
\end{thm}

\begin{proof}
We store the following information:
\begin{itemize}
\item For each clan $C\neq V(H)$ with $m$ ghosts and each $m$-tuple of vertices $\ivl{w}{1}{m}$ of $G$
we record the number $$S(C,\ivl{w}{1}{m})=\sum_{v\in N^+_1(w_1)} \hom_{(H,T)}(C,v,\ivl{w}{1}{m}),$$
that is the number of homomorphisms of $C^\ast$ to $G$ such that the ghosts of $C$ map to $\ivl{w}{1}{m}$ and $r(C)$
maps to some outneighbor $v$ of $w_1$.
\item For each $v\in V(G)$, the number $\hom((H,T),G,v)$.
\item The sum $\hom(H,G)$ of these numbers over all vertices of $G$.
\end{itemize}
The number $S(C,\ivl{w}{1}{m})$ is only stored for those combinations of $C$ and $\ivl{w}{1}{m}$ for that it is non-zero.
The values are stored in a hash table (see e.g.~\cite{4algs} for implementation details), so that they can be accessed in a constant time.
By Lemma~\ref{lem:clans}, if $\hom_{(H,T)}(C,v,\ivl{w}{1}{m})$ is non-zero, then $w_1$, \ldots, $w_m$ are in-neighbors of
$v$ in $G$.  Since the maximum indegree of $G$ is at most $D$, each vertex $v$ contributes at most $D^{|V(H)|}$ non-zero values
(and each of the numbers is smaller or equal to $|V(G)|^{|V(H)|}$), thus the space necessary for the storage is $\Oh(D^{|V(H)|}|V(G)|)$.
Queries can be performed in a constant time by returning the stored information.

The addition of an edge $(x,y)$ to $G$ is implemented as follows.  We process the clans of $(H,T)$
in the decreasing order of size, i.e., when we use the information stored for the smaller clans,
it still refers to the graph $G$ without the new edge.  Let us consider a clan $C\neq V(H)$ with ghosts $\ivl{g}{1}{m}$.
For each non-empty set $X$ of edges of $C^\ast$ which have the same color as $(x,y)$, we are going to find
all vertices $v$ and $\ivl{w}{1}{m}$ such that there exists a homomorphism of $C^\ast$ mapping
$r(C)$ to $v$ and the ghosts of $C$ to $\ivl{w}{1}{m}$ which maps precisely the edges of $X$ to $(x,y)$.
We will also determine the numbers of such homomorphisms, and decrease the number $S(C,\ivl{w}{1}{m})$
by this amount.  Note that the number of choices of $X$ is constant (bounded by a function of $H$).

Consider now a fixed set $X$.  Let $M$ be the set of vertices $z\in V(C^\ast)$
such that there exists a directed path in $C^\ast$ from $z$ to the head of an edge of $X$.
Note that $r(C)$ and all ghosts of $C$ belong to $M$.  Let $C_1$, \ldots, $C_t$ be the
vertex sets of connected components of $C^\ast-M$, and observe that they are clans.
Now, let ${\cal F}$ be the set of all homomorphisms from the subgraph of $C^\ast$ induced by $M$ to $G+(x,y)$
such that exactly the edges of $X$ are mapped to $(x,y)$.  Note that if $z$ is an image of a vertex of $M$
in such a homomorphism, then $G$ contains a directed path from $z$ to $y$ of length at most $|V(H)|$,
thus there are only $\Oh(D^{|V(H)|})$ vertices of $G$ to that $M$ can map, and consequently only
$\Oh(D^{|V(H)|^2})$ choices for the homomorphisms.  Each such choice fixes the image of $r(C)$ as well
as all the ghosts.

Consider $\phi\in {\cal F}$.  We need to determine in how many ways $\phi$
extends to a homomorphism of $C^\ast$ that maps no further edges to $(x,y)$ (this number is then
added to the value $S(C,\phi(g_1),\ldots, \phi(g_m))$).  Note that for $1\le i\le t$,
the ghosts of $C_i$ are contained in $M$, and thus their images are fixed by the choice of $\phi$.
Therefore, if $\ivl{g^i}{1}{m_i}$ are the ghosts of $C_i$, then the number of the homomorphisms extending $\phi$
is $$\prod_{i=1}^t S(C_i,\phi(g^i_1),\ldots, \phi(g^i_{m_i})).$$
Here, we use the fact that the values $S(C_i,\ldots)$ were not updated yet, and thus in the homomorphisms
that we count, no other edge maps to $(x,y)$.  These products can be determined in a constant time.

The values $\hom((H,T),G,v)$ are updated similarly, before the values $S(C,\ldots)$ are updated.
The changes in the values of $\hom((H,T),G,v)$ are also propagated to the stored value of $\hom(H,G)$.
The complexity of the update is given by the number of choices of partial homomorphisms ${\cal F}$,
i.e., $\Oh(D^{|V(H)|^2})$.

Edge removal works in the same manner, except that the information is subtracted in the end, and
that the clans are processed in the opposite direction, i.e., starting from the inclusion-wise smallest clans,
so that the values for the graph without the edge are used in the computations.

Change of the orientation of an edge or its recoloring can be implemented as subsequent deletion and addition.
The data structure can be initialized by adding edges one by one, starting with the data structure for
an empty graph $G$ whose initialization is trivial.
\end{proof}

\subsection{Induced subgraphs}\label{subsec-complex}

The data structure $\ISub_{H,k}(G)$ consists essentially of the data structure for maintaining
the $h$-th augmentation $G'$ of $G$ (where $h=\binom{|V(H)|}{2}-2$) and of a constant (bounded by a function of $k$ and $H$) number
of data structures $\AHom$, to that we have to propagate all the changes in $G'$.  Therefore,
if $D$ is the bound from the data structure from Theorem~\ref{thm:aug}, then
the data structure $\ISub_{H,k}(G)$ has the following complexities (amortized).
\begin{itemize}
\item Edge addition: $\Oh(D^{|V(H)|^2+1}\log^{\binom{|V(H)}{2}-1} |V(G)|)$.
\item Edge removal: $\Oh(D^{|V(H)|^2+1})$.
\item Edge recoloring: $\Oh(D^{|V(H)|^2})$.
\item Initialization: $\Oh(D^{|V(H)|^2+1}|V(G)|+t)$.
\item Space: $\Oh(D^{|V(H)|^2}|V(G)|)$.
\end{itemize}

For dense graphs, the bound on $D$ is too large for the data structure to be useful.
However, if $G$ is kept within some class $\G$ of graphs with bounded expansion, then the function
$h(n,r)$ from Theorem~\ref{thm:aug} can be chosen to be constant, and we obtain $D$ constant.
Therefore, when applied to such a class of graphs, the complexities are as follows.
\begin{itemize}
\item Edge addition: $\Oh(\log^{\binom{|V(H)}{2}-1} |V(G)|)$.
\item Edge removal and recoloring: $\Oh(1)$.
\item Initialization: $\Oh(|V(G)|)$.
\item Space: $\Oh(|V(G)|)$.
\end{itemize}

Similarly, if $\G$ is nowhere-dense, then the function $h(n,r)$ is $\Oh(n^{\eps'})$ for
any fixed $r$ and any $\eps'>0$.  Therefore, given any $\eps>0$, we can choose $\eps'$ to be
less than $\eps/(|V(H)|^2+1)$ and the complexities of the data structure are as follows.

\begin{itemize}
\item Edge addition, removal and recoloring: $\Oh(|V(G)|^\eps)$.
\item Initialization and space: $\Oh(|V(G)|^{1+\eps})$.
\end{itemize}

\section{Extensions} \label{sec:logi}
Although we have for simplicity formulated the data structure $\ISub$ for a graph $G$ with
a fixed vertex set, there is no problem with adding or removing isolated vertices to/from $G$
in a constant time.

One can ask about a number of possible extensions to the data structure $\ISub$.  Can we allow
directed edges? Or colors of vertices? Or hyperedges?
All these can be expressed as relational structures.
By \emph{dictionary} $\sigma$ we mean a finite set of symbols along with finite arities.
\emph{Relational structure} $S$ for a given dictionary $\sigma$ is a set $V(S)$ called
the \emph{universe} of $S$ along with the realization of the relations, i.e., for every $R$ of
arity $r$ in $\sigma$ the set $R^S\subseteq V(S)^r$ of $r$-tuples that satisfy the relation.
As an example, a graph is a relational structure with a single binary relation that
is satisfied for the pairs of vertices joined by edges.  Colors of vertices and edges can be represented
by additional unary and binary relations, respectively.

We define $|S|$ as $|V(S)| + \sum_{R\in \sigma} |R^S|$.
For $\ivl{x}{1}{r} \in V(S)^r$ we use the shorthand $R(\ivl{x}{1}{r})$ for $(\ivl{x}{1}{r})\in R^S$ when $S$ is clear from the context.
The \emph{Gaifman graph} of a structure $S$ is the undirected graph $G_S$ with $V(G_S) = V(S)$ and an edge between two distinct vertices $a,\ b\in V(G_S)$ if there exist a relation $R\in \sigma$ of arity $r$ and a tuple $\ivl{x}{1}{r}\in V(S)^{r}$ such that $R(\ivl{x}{1}{r})$ and $a,\ b\in \ivl{x}{1}{r}$.
The \emph{incidence graph} of a structure $S$ with dictionary $\sigma$ is the bipartite undirected graph $G^i_S$
with edges colored by $\sigma\cup \{c\}$, where $c$ is a color not contained $\sigma$, defined as follows.
The vertex set consists of $V(S)$ and of vertices $v_{R,\ivl{x}{1}{k}}$ and $v'_{R,\ivl{x}{1}{k}}$
for each $k$-ary relational symbol $R$ and $k$-tuple $\ivl{x}{1}{k}$ such that $R(\ivl{x}{1}{k})$ holds.
For each such $R$ and $k$-tuple, $G^i_S$ contains an edge $v_{R,\ivl{x}{1}{k}}v'_{R,\ivl{x}{1}{k}}$ colored by $R$
and edges $x_iv_{R,\ivl{x}{1}{k}}$ for $1\le i\le k$ colored by $c$.

A class of structures $\mathcal{S}$ is said to have bounded expansion (be nowhere dense), if the
class $\{G_S;\ S\in \mathcal{S}\}$ has bounded expansion (is nowhere dense, respectively).
Note that in such a case, the maximum average degree of $G_S$ is bounded by a constant ($|V(S)|^\eps$ for every $\eps>0$,
respectively).  Consequently, the number of cliques in $G_S$ of size bounded by the maximum arity of
a relation symbol of $S$ is $\Oh(|V(S)|)$ ($\Oh(|V(S)|)^{1+\eps}$ for every $\eps>0$, respectively), see~\cite{wood-cliques}.
It follows that we have $|S| = \Oh(|V(G)|)$ ($\Oh(|V(S)|)^{1+\eps}$ for every $\eps>0$, respectively).

\begin{lemma}
If a class of structures $\mathcal{S}$ has bounded expansion (is nowhere-dense), then the class $\mathcal{S}^i = \{G_S^i;\ S\in \mathcal{S}\}$ has bounded expansion (is nowhere dense, respectively).
\end{lemma}
\begin{proof}
We present the proof for bounded expansion.  The argument for the nowhere-dense case is analogical.

For $t>0$ and a class $\G$ of graphs, let $\G\tnabla t$ denote the set of all graphs $G$ such that
there exists a graph $G'\in \G$ and a graph obtained from $G$ by subdividing each edge at most $t$
times is a subgraph of $G'$.  As was shown in~\cite{subdivchar}, $\G$ has bounded expansion
if and only if for every $t>0$, there exists a constant $c_t$ such that all graphs in $\G\tnabla t$
have average degree at most $c_t$.

Suppose that there exists $t$ such that we can find arbitrarily dense graphs in $\mathcal{S}^i\tnabla t$.
Since $\mathcal{S}^i\tnabla t$ is closed on subgraphs, it also contains graphs of arbitrarily large minimum degree.
Let $k$ be the maximum arity of the symbols in the dictionary of $S$ and
let $H$ be a graph in $\mathcal{S}^i\tnabla t$ with minimum degree at least $k+2$,
and let $S\in\mathcal{S}$ be a relational structure such that a graph $H'$ obtained from $H$ by subdividing each edge
at most $t$ times appears as a subgraph of $G_S^i$.

The branching vertices of $H'$ in $G_S^i$ must correspond to the vertices of $V(S)$,
since all other vertices of $G_S^i$ have degree at most $k+1$.
However, for every path of length two in $G_S^i$ between two vertices $u$ and $v$ corresponding to vertices of $V(S)$,
there is an edge in $G_S$ between $u$ and $v$.  We conclude that $H\in \{G_S;\ S\in \mathcal{S}\} \tnabla \lceil t/2\rceil$.
Therefore, $\{G_S;\ S\in \mathcal{S}\} \tnabla \lceil t/2\rceil$ would contain graphs of arbitrarily large minimum degree,
contradicting the assumption that $\mathcal{S}$ has bounded expansion.
\end{proof}

To count the number of appearances of a fixed relational structure $S_0$ as an induced substructure of a relational structure $S$,
it suffices to count the number of appearances of $G^i_{S_0}$ in $G^i_S$ as an induced subgraph
(assuming that every vertex of $S_0$ belongs to at least one relation; the case that $S_0$ contains isolated vertices
can be dealt with by introducing a new unary relation satisfied for all vertices).
A change (addition or removal of a tuple to/from a relation) in $S$ results in only a constant number of changes in $G^i_S$, thus
$\ISub$ can be used to represent relational structures through this transformation.

A seemingly more general question is testing existential first order properties,
i.e., properties which can be defined by closed first-order formulas using only non-negated existential quantifiers.
Such a formula $\phi$ can be considered to be in the disjunctive normal form, i.e.
$$ \phi = \bigvee_{i=1}^t \ivl{\exists x}{1}{l} \phi_i, $$
where each $\phi_i$ is a conjunction of a finite number of terms of the form
$R(\ivl{x}{{i_1}}{{i_k}})$ or $\neg R(\ivl{x}{{i_1}}{{i_k}})$ for a relation $R\in \sigma\cup \{=\}$ of arity $k$.
For example, existence of an induced subgraph, subgraph or homomorphism from a graph $H$ of a bounded size can be expressed this way,
by a formula having one variable for each vertex of $H$ and describing the required adjacency, non-adjacency and non-equality
relations between them.

In order to decide whether a structure $S$ satisfies the given formula $\phi$, we need to find a set of witnessing vertices $\ivl{x}{1}{l}$ which satisfies the subformula $\phi_i$ for some $1\le i\le t$.
For every such $\phi_i$ there is only a finite number of structures $S_i$ on at most $l$ vertices which satisfy $\phi_i$.
Hence, it suffices to check whether one of these structures is an induced substructure of $S$.
Therefore, using the data structure $\ISub$, we can decide arbitrary existential first order properties with a bounded
number of variables on classes of structures with bounded expansion (or nowhere-dense), within the same time bounds.

\section{Concluding remarks}\label{sec:conc}

A natural question is whether one can design a fully dynamic data structure to
decide properties expressible in First Order Logic on graphs with bounded expansion.
For this purpose, it would be convenient to be able to maintain low tree-depth colorings of~\cite{npom-old},
which however appears to be difficult.

Possibly a much easier problem is the following.  We have described a dynamic data structure
that enables us to count the number of appearances of $H$ as an induced subgraph of $G$,
for graphs from a class with bounded expansion.  If this number is non-zero, can we find such
an appearance?  Getting this from our data structure is not entirely trivial, due to the use of
the principle of inclusion and exclusion.

By a famous result of Courcelle, any property expressible in Monadic Second Order Logic
can be tested for graphs of bounded tree-width in linear time.  Can one design
a dynamic data structure for this problem?  It is not even clear how to
maintain a tree decomposition of bounded width dynamically.

\bibliographystyle{siam}
\bibliography{subg_in_gwbe}

\end{document}